\documentclass[journal]{IEEEtran}
\usepackage{amsmath}
\usepackage{graphicx}
\usepackage{amsfonts}
\usepackage{booktabs}
\usepackage{url}
\usepackage{cite}
\usepackage{amssymb}
\usepackage{diagbox}
\usepackage{subfigure}
\usepackage{comment,soul}
\usepackage{color}
\usepackage{multirow}
\usepackage{mathtools}
\usepackage{algorithmic}
\usepackage{algorithm}
\usepackage{threeparttable}
\usepackage{tablefootnote}

\usepackage{amsthm}
\theoremstyle{definition}
\newtheorem{theorem}{Theorem}

\newtheorem{prop}[theorem]{Proposition}
\newsavebox\CBox
\def\textBF#1{\sbox\CBox{#1}\resizebox{\wd\CBox}{\ht\CBox}{\textbf{#1}}}

\title{Rethinking Processing Distortions: Disentangling the Impact of Speech Enhancement Errors on Speech Recognition Performance}
\author{Tsubasa Ochiai~\IEEEmembership{Member,~IEEE}, Kazuma Iwamoto,\\ Marc Delcroix~\IEEEmembership{Senior Member,~IEEE}, Rintaro Ikeshita~\IEEEmembership{Member,~IEEE},\\ Hiroshi Sato~\IEEEmembership{Member,~IEEE}, Shoko Araki~\IEEEmembership{Fellow,~IEEE}, Shigeru Katagiri~\IEEEmembership{Fellow,~IEEE}
\thanks{Tsubasa Ochiai, Marc Delcroix, Rintaro Ikeshita, Hiroshi Sato, and Shoko Araki are with NTT Corporation, Chiyoda-ku 100-8116, Japan (email: tsubasa.ochiai@ntt.com; marc.delcroix@ntt.com; rintaro.ikeshita@ntt.com; hrs.sato@ntt.com; shoko.araki@ntt.com).

Kazuma Iwamoto and Shigeru Katagiri are with Doshisha University, Kyotanabe 610-0394, Japan (kazuma.iwamoto1204@gmail.com; skatagir@mail.doshisha.ac.jp).}}

\begin{document}
\maketitle
\begin{abstract}

It is challenging to improve automatic speech recognition (ASR) performance in noisy conditions with a single-channel speech enhancement (SE) front-end.
This is generally attributed to the processing distortions caused by the nonlinear processing of single-channel SE front-ends.
However, the causes of such degraded ASR performance have not been fully investigated.
How to design single-channel SE front-ends in a way that significantly improves ASR performance remains an open research question.
In this study, we investigate a signal-level numerical metric that can explain the cause of degradation in ASR performance.
To this end, we propose a novel analysis scheme based on the orthogonal projection-based decomposition of SE errors.
This scheme manually modifies the ratio of the decomposed interference, noise, and artifact errors, and it enables us to directly evaluate the impact of each error type on ASR performance.
Our analysis reveals the particularly detrimental effect of artifact errors on ASR performance compared to the other types of errors.
This provides us with a more principled definition of processing distortions that cause the ASR performance degradation.
Then, we study two practical approaches for reducing the impact of artifact errors.
First, we prove that the simple observation adding (OA) post-processing (i.e., interpolating the enhanced and observed signals) can monotonically improve the signal-to-artifact ratio.
Second, we propose a novel training objective, called artifact-boosted signal-to-distortion ratio (AB-SDR), which forces the model to estimate the enhanced signals with fewer artifact errors.
Through experiments, we confirm that both the OA and AB-SDR approaches are effective in decreasing artifact errors caused by single-channel SE front-ends, allowing them to significantly improve ASR performance.

\end{abstract}

\IEEEpeerreviewmaketitle
\section{Introduction}
\label{sec:introduction}


Over the past decade, the performance of automatic speech recognition (ASR) systems has progressively improved \cite{hinton2012deep,yu2016automatic,li2022recent} with the advent of deep learning techniques \cite{lecun2015deep,goodfellow2016deep}.
Research interest has thus moved toward achieving robustness against acoustic interference such as background noise, reverberation, and overlapping speakers \cite{haeb2020far}.
Recent studies and ASR challenges \cite{barker2015third,barker2018fifth} have reported that multichannel array processing, i.e., mask-based beamformer \cite{higuchi2016robust,heymann2016neural,boeddeker2018exploring,boeddeker2018front}, plays a key role in improving robustness to interference, and it has been used as the de facto standard front-end for noise-robust ASR systems.
However, it is not always possible to equip commercial devices with microphone arrays due to structural constraints and cost restrictions.
Therefore, building noise-robust ASR systems with a single microphone (i.e., single-channel) is an essential research issue.

Recent advances in deep learning have greatly improved the performance of single-channel speech enhancement (SE), such as noise reduction \cite{wang2018supervised, pandey2019tcnn}, source separation \cite{wang2018supervised,luo2019conv}, and target speech extraction \cite{zmolikova2023neural}.
Single-channel SE approaches can improve SE performance measures, such as signal-to-distortion ratio (SDR) \cite{vincent2006performance}, perceptual evaluation of speech quality (PESQ) \cite{rix2001perceptual}, and short-time objective intelligibility (STOI) \cite{taal2011algorithm}.
However, previous studies (e.g., \cite{yoshioka2015ntt,chen2018building,fujimoto2019one,menne2019investigation,chang2022end-to-end}) have reported that such single-channel SE approaches do not necessarily contribute to improving the ASR performance for the ASR systems trained with multi-condition training data \cite{vincent2017analysis} (i.e., they tend to only slightly improve or even degrade word error rate (WER)).

A common view in the ASR research community is that \emph{processing distortions} produced by single-channel nonlinear processing, like neural networks, induce such ASR performance degradation.
Researchers have focused on developing approaches to mitigate the effects of such potential distortions.
Such approaches include retraining the ASR back-end on enhanced speech (e.g., \cite{chen2018building}) or joint training of the SE front-end and ASR back-end (e.g., \cite{menne2019investigation,chang2022end-to-end,woo2020end,von2020multi,wang2016joint}).
However, it is not always possible to tune the ASR back-end for a specific SE front-end, due to the need to use an already deployed ASR back-end and the cost of training and maintaining an ASR system for different SE front-ends, etc.
Therefore, modular ASR systems are often preferable in practice, and such systems require the design of effective standalone SE front-ends. 

How to design single-channel SE front-ends that can significantly improve ASR performance remains an open research question.
We believe that a key to answering this question is a better understanding of the nature of processing distortions that are detrimental to ASR.
Surprisingly, there have been no detailed systematic analyses or interpretations of the processing distortions and their impact on ASR.
We aim to fill this gap in this study.
In particular, we investigate a signal-level numerical metric that could explain, at least partly, the cause of ASR performance degradation.

\subsection{Hypotheses on cause of ASR performance degradation}

An SE model estimates the target clean source given the observed noisy signal.
The processing distortions can be measured as the total estimation errors between the reference source signal and the enhanced signal.
We hypothesize that the total estimation errors can be decomposed into \emph{natural} and \emph{unnatural} error components.
Natural errors would consist of residual interference and noise signals, which are similar to noisy samples seen when training the ASR back-end with multi-condition data.
Consequently, such natural (residual) errors may have little impact on ASR performance.
In contrast, unnatural errors consist of processing artifacts caused by single-channel nonlinear processing.
Such unnatural (processing) errors may have a more detrimental effect on ASR because it is difficult to train an ASR back-end that is robust to these artifacts due to their large diversity.
A signal-level numerical metric is required to separate the total estimation error into natural and unnatural error components.

To define the error components, we borrow the orthogonal projection-based decomposition of the SE estimation errors, which was previously proposed as an SE performance metric for speech denoising and separation, such as the SDR of BSSEval \cite{vincent2006performance}.
It decomposes the SE errors into three error components, i.e., the residual \emph{interference} error, the residual \emph{noise} error, and the \emph{artifact} error.
The interference error component consists of a linear combination of target speech and interference signals, while the noise error component consists of a linear combination of target speech, interference, and noise signals.
Accordingly, the interference and noise errors represent the signals that could be observed in the real world, i.e., natural signals.
On the other hand, the artifact error component consists of a signal orthogonal to the target speech, interference, and noise signals (Figure~\ref{fig:decomposition}-(a)); in other words, it cannot be represented by a linear combination of target speech, interference, and noise signals.
Thus, the artifact errors represent the signals that are produced by nonlinear processing with single-channel SE, which can be considered unnatural (artificial) signals.
Considering the above observations, we \emph{hypothesize} that the artifact errors have a larger impact on the degradation of ASR performance than on the other residual interference and noise errors.
In other words, we assume that the artifact errors could mostly explain the limited ASR performance improvement induced by single-channel SE.

\subsection{Contributions: Experimental proof of impact of artifact errors and practical approaches to mitigate them}

In this study, to confirm this hypothesis, we propose a novel analysis scheme that decomposes the SE errors by utilizing the reference signals (i.e., target speech, interfering speech, an background noise), manually modifies the ratio of the interference, noise, and artifact errors included in the estimated signals, and directly evaluates the ASR performance of the modified signals.
We refer to the proposed analysis scheme as direct scaling analysis (DSA).
This analysis enables us to directly investigate the impact of each SE error on the ASR performance and proves the particularly negative impact of artifact errors on ASR performance as well as the relatively limited impact of the interference and noise errors.

The results of this analysis motivate us to explore two practical approaches to mitigate the effect of the artifact errors generated by the single-channel SE for improving the ASR performance: 1) observation adding (OA) and 2) artifact-boosted signal-to-distortion ratio (AB-SDR) loss.

The OA post-processing interpolates the enhanced signal and the observed noisy signal, and it is applied in the inference stage as the post-processing of the single-channel SE.
In the literature, it has often been used to reduce the auditory nonlinear distortion (i.e., musical noise \cite{cappe1994elimination}) of the enhanced signals generated by classical spectral subtraction \cite{boll1979suppression} and binary masking \cite{lyon1983computational}.
However, to the best of our knowledge, OA has not been investigated as an approach to improving the ASR performance of recent single-channel SE systems.
Moreover, its effect on decomposed SE errors has not been considered.
In this paper, we mathematically and experimentally show that the OA post-processing has an effect on reducing the ratio of artifact errors, which contributes to improving ASR performance based on our hypothesis.

We also propose a novel training objective for an SE front-end, i.e., AB-SDR loss, which extends the conventional SDR loss to put more weight on the artifact errors than on the noise and interfering errors.
Unlike the OA post-processing, AB-SDR loss is applied in the training stage as the training objective.
We experimentally show that AB-SDR successfully reduces the ratio of artifact errors and improves the ASR performance of single-channel SE.

Although the two approaches try to reduce the artifact errors in different ways, we show experimentally that both of them successfully reduce artifact errors at the cost of increasing other types of errors less harmful to ASR and consequently improve ASR performance.
These experiments also provide experimental evidence of our hypothesis that the (unnatural) artifact errors have a larger impact on the degradation of ASR performance compared to the other (natural) residual interference and noise errors.

The main contributions of this paper can be summarized as follows:
\begin{enumerate}
    \item We propose a novel analysis scheme (i.e., DSA) to confirm the validity of our hypothesis that the artifact errors should be one of the reasonable signal-level numerical metrics representing the cause of ASR performance degradation, which has not been explored in the noise-robust ASR community.
    \item The experimental analyses provide a novel insight into the causes of the limited ASR performance improvement induced by single-channel SE: 1) artifact errors have a larger impact on ASR than interference and noise errors and 2) ASR performances of single-channel SE could be improved by mitigating artifact errors.
    \item We propose two practical approaches to mitigate the impact of the artifact errors: 1) the OA post-processing and 2) the AB-SDR training objective.
    The experimental results show the effectiveness of these approaches in improving the ASR performance of single-channel SE, even for real recordings.
\end{enumerate}

This paper extends our previous study \cite{iwamoto2022how} by 1) generalizing the derivation and proof from a single-talker to a multi-talker scenario, 2) proposing the AB-SDR training objective as a more straightforward approach to reducing the ratio of artifact errors, and 3) performing comprehensive experiments including different acoustic conditions, evaluation corpora, and SE/ASR systems.

The remainder of this paper is summarized as follows.
In Section~\ref{sec:related}, we briefly discuss related prior works.
Section~\ref{sec:background} overviews the orthogonal projection-based decomposition of the SE errors.
Section~\ref{sec:proposed} explains the details of 1) the proposed DSA scheme, 2) the OA approach, and 3) the AB-SDR approach.
In Sections~\ref{sec:exp_cond}--\ref{sec:exp_add}, we explain the experiments of the single-channel SE/ASR setups in noisy environments and give novel insights into the processing distortions.
Finally, we conclude this paper in Section~\ref{sec:conclusion}.

\section{Related works}
\label{sec:related}

\subsection{Observation adding}

Concurrently with our preliminary study \cite{iwamoto2022how}, the authors of another study \cite{zorila2022speaker} proposed using the OA post-processing to alleviate the processing distortions introduced by an SE front-end and reported an improvement in ASR performance.
However, they did not analyze how OA post-processing could improve ASR performance.

Several studies \cite{wang2020voicefilter,sato2021should,sato2022learning} have investigated approaches to reducing the effect of processing distortion on noisy multi-speaker mixtures.
They hypothesized it would be better to recognize noisy speech directly when no interference speaker was active and enhanced speech otherwise.
They proposed switching the input of the ASR back-end between the enhanced and observed signals according to the interference and noise conditions.
As a formalization, they adopted the weighted sum of the enhanced and observed signals, which is similar to that of the OA post-processing, where the switching weights are dynamically computed based on the overlap detector or switcher.
They showed the effectiveness of such switching schemes on ASR performance.
However, the effectiveness of the OA post-processing alone (i.e., simply interpolating the enhanced and observed signals with a pre-defined interpolation weight) on ASR performance has not been fully proven.

In contrast to these prior works, this study offers a hypothesis on the mathematical definition for the cause of ASR performance degradation and experimentally confirms the validity of our hypothesis with the proposed novel analysis scheme.

\subsection{Artifact-aware training objective}

The authors of a previous study \cite{boeddeker2021convolutive} proposed using the SDR of BSSEval \cite{vincent2006performance} as the training objective for neural network-based SE models.
Because the SDR loss minimizes the total estimation errors (i.e., the sum of interference, noise, and artifact errors), it does not allow balancing the contribution of each error type within the loss.

Several studies \cite{venkataramani2018performance,koizumi2022snri} have investigated adding the signal-to-artifact ratio (SAR) \cite{vincent2006performance} to their training objective.
The former study \cite{venkataramani2018performance} focused on evaluation in terms of subjective listening quality but did not investigate the impact of artifact errors on ASR performance.
The latter study \cite{koizumi2022snri} investigated the joint training framework with the ASR-level training objective and adopted the SAR as an additional regularization term.
However, the effectiveness of the SAR loss alone has not been confirmed, and thus it has not explicitly shown how the SAR loss without the ASR-level loss contributes to improving ASR performance.
In particular, in our preliminary experiment, we observed that the SAR-based loss may cause training instabilities probably because the SAR becomes optimal when the estimated signal matches the observed noisy signal.
This may indicate that it is challenging to train an SE system without the ASR-level loss.

In contrast to these prior works, this study proposes a novel analysis scheme and, for the first time, explicitly shows that artifact errors negatively impact ASR performance.
Furthermore, rather than adding the SAR-based regularization, we introduce the AB-SDR loss, which is a signal-level training objective that boosts the effect of the artifact errors in the SDR loss and can improve ASR performance without relying on the ASR-level training objective.

\subsection{ASR-level training objectives}

As described in Section~\ref{sec:introduction}, the joint training of the SE front-end and ASR back-end (e.g., optimizing the SE front-end with the ASR-level training objective) would be a straightforward approach to mitigating the processing distortion problem, which could be done without deeply understanding its mechanism.
Although such a joint training scheme successfully improves ASR performance (e.g., \cite{menne2019investigation,chang2022end-to-end,woo2020end,von2020multi,wang2016joint}), it requires modifying the SE front-end and/or the ASR back-end, which is not always an option in practice, e.g., due to the need to use an already deployed system.
Moreover, its effect on SE estimation errors has not been fully explored; the previous studies \cite{woo2020end,von2020multi} confirmed the improvements in the ASR metric but not in SE metrics such as SDR.

\section{Background}
\label{sec:background}

\subsection{Overview of speech enhancement task}

In this study, we focus on the single-channel SE task, where the observed signals are contaminated by the interfering speech and background noise signals.
Let $\mathbf{y} \in \mathbb{R}^{T}$ denote a vector of a $T$-length time-domain waveform of the observed signals.
The observed noisy signals can be modeled as:
\begin{align}
    \mathbf{y} = \mathbf{s} + \mathbf{i} + \mathbf{n},
    \label{eq:obs}
\end{align}
where, $\mathbf{s} \in \mathbb{R}^{T}$ denotes the target speaker's speech signal, $\mathbf{i} \in \mathbb{R}^{T}$ denotes the interfering speaker's speech signal, and $\mathbf{n} \in \mathbb{R}^{T}$ denotes the background noise signal.
In this paper, to simplify the description, we describe the equations assuming a single interfering speaker.

The goal of the SE task is to extract the target speaker's speech signals from the observed noisy signals, i.e., to reduce the interfering speech and background noise signals.
Given the observed signal $\mathbf{y}$ as an input, the SE system estimates the target signal $\widehat{\mathbf{s}} \in \mathbb{R}^{T}$ as:
\begin{align}
    \widehat{\mathbf{s}} = \text{SE}(\mathbf{y}),
\end{align}
where $\text{SE}(\cdot)$ denotes the functional representation of SE processing.

One of the major uses of an SE system is as a front-end for ASR to improve performance in noisy conditions.
Given the enhanced signal as an input, the ASR system estimates the word sequence $\widehat{\mathbf{W}} \in \mathcal{V}^{N}$, where $\mathcal{V}$ denotes a set of symbols and $N$ denotes the sequence length, as:
\begin{align}
    \widehat{\mathbf{W}} = \text{ASR}(\widehat{\mathbf{s}}),
\end{align}
where $\text{ASR}(\cdot)$ denotes the functional representation of ASR processing.

\subsection{Orthogonal projection-based decomposition of speech enhancement errors}
\label{sec:opd}

The enhanced signal inevitably contains estimation errors.
To define SE evaluation measures such as SDR, a previous study \cite{vincent2006performance} proposed decomposing the estimation errors into three error components, i.e., interference, noise, and artifact errors.
Given the reference source, interference, and noise signals ($\mathbf{s}$, $\mathbf{i}$, and $\mathbf{n}$), the estimated signal $\widehat{\mathbf{s}}$ is decomposed as:
\begin{align}
    \widehat{\mathbf{s}} = \mathbf{s}_{\text{target}} + \mathbf{e}_{\text{interf}} + \mathbf{e}_{\text{noise}} + \mathbf{e}_{\text{artif}},
    \label{eq:decomp}
\end{align}
where $\mathbf{s}_{\text{target}} \in \mathbb{R}^{T}$ is called a target source component, and $\mathbf{e}_{\text{interf}} \in \mathbb{R}^{T}$, $\mathbf{e}_{\text{noise}} \in \mathbb{R}^{T}$, and $\mathbf{e}_{\text{artif}} \in \mathbb{R}^{T}$ denote the interference, noise, and artifact error components, respectively.
Figure~\ref{fig:decomposition}-(a) illustrates the signal decomposition, where, for illustration simplicity, we assume that there is no interference signal in the observed signal, i.e., $\mathbf{y} = \mathbf{s} + \mathbf{n}$, and that the signal delay is not considered.

\begin{figure}[t]
  \centering
  \includegraphics[width=1.0\linewidth]{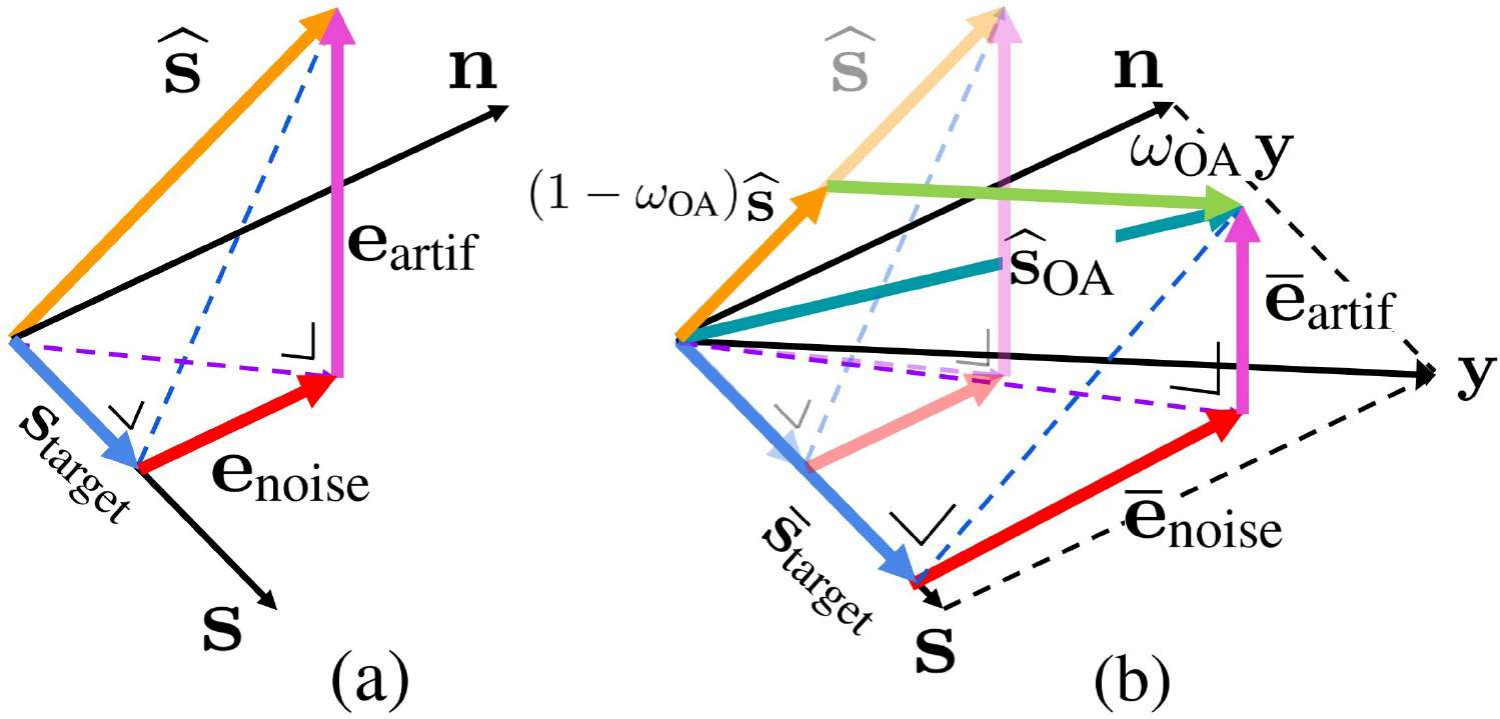}
  \caption{Illustrations of (a) orthogonal projection-based decomposition and (b) effect of observation adding from viewpoint of orthogonal projection-based decomposition.}
  \label{fig:decomposition}
\end{figure}

The above decomposition in Eq.~\eqref{eq:decomp} is defined using orthogonal projections.
Let $\mathbf{s}^{\tau}$, $\mathbf{i}^{\tau}$, and $\mathbf{n}^{\tau}$ be the source, interference, and noise signals delayed by $\tau$ samples.
Three orthogonal projection matrices are defined for the decomposition: 1) orthogonal projection matrix onto the subspace spanned by the source signals $\{ \mathbf{s}^{\tau} \}_{\tau = 0}^{L-1}$ (denoted by $\mathbf{P}_{\mathbf{s}} \in \mathbb{R}^{T \times T}$), 2) orthogonal projection matrix onto the subspace spanned by the source and interference signals $\{ \mathbf{s}^{\tau}, \mathbf{i}^{\tau} \}_{\tau = 0}^{L-1}$ (denoted by $\mathbf{P}_{\mathbf{s}, \mathbf{i}} \in \mathbb{R}^{T \times T}$), and 3) orthogonal projection matrix onto the subspace spanned by the source, interference, and noise signals $\{ \mathbf{s}^{\tau}, \mathbf{i}^{\tau}, \mathbf{n}^{\tau} \}_{\tau = 0}^{L-1}$ (denoted by $\mathbf{P}_{\mathbf{s}, \mathbf{i}, \mathbf{n}} \in \mathbb{R}^{T \times T}$), where $L-1$ is the number of maximum delay allowed or, in other words, $L$ is the number of basis vectors for the projection.
These matrices can be obtained as follows:
\begin{align}
    \mathbf{P}_{\mathbf{s}} &\coloneqq \mathbf{A}_{\mathbf{s}}(\mathbf{A}_{\mathbf{s}}^{\mathsf{T}} \mathbf{A}_{\mathbf{s}})^{-1} \mathbf{A}_{\mathbf{s}}^{\mathsf{T}}, \\
    \mathbf{P}_{\mathbf{s},\mathbf{i}} &\coloneqq \mathbf{A}_{\mathbf{s}, \mathbf{i}}(\mathbf{A}_{\mathbf{s}, \mathbf{i}}^{\mathsf{T}} \mathbf{A}_{\mathbf{s}, \mathbf{i}})^{-1} \mathbf{A}_{\mathbf{s}, \mathbf{i}}^{\mathsf{T}}, \\
    \mathbf{P}_{\mathbf{s},\mathbf{i},\mathbf{n}} &\coloneqq \mathbf{A}_{\mathbf{s},\mathbf{i},\mathbf{n}}(\mathbf{A}_{\mathbf{s},\mathbf{i},\mathbf{n}}^{\mathsf{T}} \mathbf{A}_{\mathbf{s},\mathbf{i},\mathbf{n}})^{-1} \mathbf{A}_{\mathbf{s},\mathbf{i}, \mathbf{n}}^{\mathsf{T}},
\end{align}
where
\begin{align}
\mathbf{A}_{\mathbf{s}} &\coloneqq [ \mathbf{s}^{0}, \ldots, \mathbf{s}^{L-1} ],\nonumber\\
\mathbf{A}_{\mathbf{s}, \mathbf{i}} &\coloneqq [ \mathbf{s}^{0}, \ldots, \mathbf{s}^{L-1}, \mathbf{i}^{0}, \ldots, \mathbf{i}^{L-1} ],\nonumber\\
\mathbf{A}_{\mathbf{s}, \mathbf{i}, \mathbf{n}} &\coloneqq [ \mathbf{s}^{0}, \ldots, \mathbf{s}^{L-1}, \mathbf{i}^{0}, \ldots, \mathbf{i}^{L-1}, \mathbf{n}^{0}, \ldots, \mathbf{n}^{L-1} ]. \nonumber
\end{align}
Then, the decomposed components in Eq.~\eqref{eq:decomp} ($\mathbf{s}_{\text{target}}$, $\mathbf{e}_{\text{interf}}$, $\mathbf{e}_{\text{noise}}$, and $\mathbf{e}_{\text{artif}}$) are defined based on the projection matrices as:
\begin{align}
    \mathbf{s}_{\text{target}} &= \mathbf{P}_{\mathbf{s}} \widehat{\mathbf{s}}, \label{eq:target} \\
    \mathbf{e}_{\text{interf}} &= \mathbf{P}_{\mathbf{s},\mathbf{i}} \widehat{\mathbf{s}} - \mathbf{P}_{\mathbf{s}} \widehat{\mathbf{s}},\label{eq:interf} \\
    \mathbf{e}_{\text{noise}} &= \mathbf{P}_{\mathbf{s},\mathbf{i},\mathbf{n}} \widehat{\mathbf{s}} - \mathbf{P}_{\mathbf{s},\mathbf{i}} \widehat{\mathbf{s}}, \label{eq:noise} \\
    \mathbf{e}_{\text{artif}} &= \widehat{\mathbf{s}} - \mathbf{P}_{\mathbf{s},\mathbf{i},\mathbf{n}} \widehat{\mathbf{s}}. \label{eq:artif}
\end{align}
Here, $\mathbf{e}_{\text{interf}}$ and $\mathbf{e}_{\text{noise}}$ represent the residual interference and noise error components that can be expressed as a linear combination of the reference signals (i.e., source, interference, noise signals).
On the other hand, $\mathbf{e}_{\text{artif}}$ represents the error component that cannot be expressed as a linear combination of the reference signals, and thus it was deemed an artifact error that is artificially generated by the SE systems.

\subsection{Evaluation measures}

The decomposition described in Section~\ref{sec:opd} was originally proposed to derive four SE evaluation metrics \cite{vincent2006performance}:
1) SDR, 2) signal-to-interference ratio (SIR), 3) signal-to-noise ratio (SNR), and 4) SAR.
These metrics are defined as follows:
\begin{align}
    \text{SDR} &\coloneqq 10 \log_{10} \frac{\| \mathbf{s}_{\text{target}} \|^{2}}{\| \mathbf{e}_{\text{interf}} + \mathbf{e}_{\text{noise}} + \mathbf{e}_{\text{artif}} \|^{2}}, \label{eq:sdr} \\
    \text{SIR} &\coloneqq 10 \log_{10} \frac{\| \mathbf{s}_{\text{target}} \|^{2}}{\| \mathbf{e}_{\text{interf}} \|^{2}}, \label{eq:sir} \\
    \text{SNR} &\coloneqq 10 \log_{10} \frac{\| \mathbf{s}_{\text{target}} + \mathbf{e}_{\text{interf}} \|^{2}}{\| \mathbf{e}_{\text{noise}} \|^{2}}, \label{eq:snr} \\
    \text{SAR} &\coloneqq 10 \log_{10} \frac{\| \mathbf{s}_{\text{target}} + \mathbf{e}_{\text{interf}} + \mathbf{e}_{\text{noise}} \|^{2}}{\| \mathbf{e}_{\text{artif}} \|^{2}}. \label{eq:sar}
\end{align}
The $\text{SDR}$ metric represents the total error (i.e., equally penalizing each error component) between reference and estimated signals, while the other metrics (i.e., $\text{SIR}$, $\text{SNR}$, $\text{SAR}$) individually evaluate each error component such as interference, noise, and artifact.

\section{Proposed analyses and approaches to reducing artifact errors}
\label{sec:proposed}

\subsection{Controlling SIR, SNR, and SAR by directly scaling error components}
\label{sec:dsa}

To reveal the impact of the different error components (i.e., interference, noise, and artifact errors) on ASR performance, we propose a novel DSA scheme.

Based on the signal decomposition in Eq.~\eqref{eq:decomp}, we generate the modified enhanced signal $\widehat{\mathbf{s}}_{\text{DSA}} \in \mathbb{R}^{T}$ by directly scaling (increasing/decreasing) the error components $\mathbf{e}_{\text{interf}}, \mathbf{e}_{\text{noise}}, \mathbf{e}_{\text{artif}}$ as:
\begin{align}
    \widehat{\mathbf{s}}_{\text{DSA}} &= \mathbf{s}_{\text{target}} + \omega_{\text{interf}}\ \mathbf{e}_{\text{interf}} + \omega_{\text{noise}}\ \mathbf{e}_{\text{noise}} + \omega_{\text{artif}}\ \mathbf{e}_{\text{artif}},\label{eq:dsa}
\end{align}
where $\omega_{\text{interf}}$, $\omega_{\text{noise}}$, and $\omega_{\text{artif}}$ are the scaling parameters that control the amounts of interference, noise, and artifact error components.
By changing the scaling parameters, we can obtain enhanced signals with different levels of each error component $\mathbf{e}_{\text{interf}}, \mathbf{e}_{\text{noise}}, \mathbf{e}_{\text{artif}}$ while retaining the same target source component $\mathbf{s}_{\text{target}}$.

Algorithm~\ref{alg:dsa} summarizes the processing flow of DSA, where $\text{OPD}(\cdot)$ denotes the functional representation of the orthogonal projection-based decomposition (i.e., Eq.~\eqref{eq:target}--\eqref{eq:artif}) and $\text{WER}(\cdot)$ denotes the functional representation of the WER (i.e., Levenshtein distance) computation.
Applying DSA assumes that the observed noisy signal $\mathbf{y}$, reference source, interference, and noise signals $(\mathbf{s}, \mathbf{i}, \mathbf{n})$, and reference transcription $\mathbf{W}$ are available.

\begin{figure}[!t]
\begin{algorithm}[H]
    \caption{Processing flow of direct scaling analysis}
    \label{alg:dsa}
    \begin{algorithmic}[1]
    \REQUIRE evaluation dataset $D$ (each element $d \in D$ is a set of the observed signal, reference signals, and reference transcription, i.e., $d = \{\mathbf{y}, \mathbf{s}, \mathbf{i}, \mathbf{n}, \mathbf{W}\}$), trained SE model $\text{SE}(\cdot)$, trained ASR model $\text{ASR}(\cdot)$
    \ENSURE set of evaluated scores $\text{Scores}$
    \STATE{$\text{Scores} \gets \emptyset$}
    \FORALL{$\{\mathbf{y}, \mathbf{s}, \mathbf{i}, \mathbf{n}, \mathbf{W}\} \in D$}
        \STATE{$\widehat{\mathbf{s}} = \text{SE}(\mathbf{y})$}
        \STATE{$\mathbf{s}_{\text{target}}, \mathbf{e}_{\text{interf}}, \mathbf{e}_{\text{noise}}, \mathbf{e}_{\text{artif}} = \text{OPD}(\widehat{\mathbf{s}}, \mathbf{s}, \mathbf{i}, \mathbf{n})$}
        \FORALL{$\omega_{\text{interf}} \in \{0.1, 0.2, \ldots, 1.5 \}$}
            \FORALL{$\omega_{\text{noise}} \in \{0.1, 0.2, \ldots, 1.5 \}$}
                \FORALL{$\omega_{\text{artif}} \in \{0.1, 0.2, \ldots, 1.5 \}$}
                    \STATE{$\hspace{-0.5mm} \widehat{\mathbf{s}}_{\text{DSA}} = \mathbf{s}_{\text{target}} + \omega_{\text{interf}}\ \mathbf{e}_{\text{interf}} + \omega_{\text{noise}}\ \mathbf{e}_{\text{noise}} + \omega_{\text{artif}}\ \mathbf{e}_{\text{artif}}$}
                    \STATE{$\hspace{-0.5mm} \widehat{\mathbf{W}}_{\text{DSA}} = \text{ASR}(\widehat{\mathbf{s}}_{\text{DSA}})$}
                    \STATE{$\hspace{-0.5mm} S_{\text{DSA}} =\text{WER}(\mathbf{W}, \widehat{\mathbf{W}}_{\text{DSA}})$}
                    \STATE{$\hspace{-0.5mm} \text{Scores} \gets \text{Scores} \cup \{(\omega_{\text{interf}}, \omega_{\text{noise}}, \omega_{\text{artif}}, S_{\text{DSA}})\}$}
                \ENDFOR
            \ENDFOR
        \ENDFOR
    \ENDFOR
    \RETURN{$\text{Scores}$}
    \end{algorithmic}
\end{algorithm}
\end{figure}

Given the observed and reference signals, we first obtain the enhanced signals $\widehat{\mathbf{s}}$ by applying the SE system to the observed noisy signal $\mathbf{y}$ (line 3) and then obtain the decomposed signals $\mathbf{s}_{\text{target}}, \mathbf{e}_{\text{interf}}, \mathbf{e}_{\text{noise}}, \mathbf{e}_{\text{artif}}$ by applying the orthogonal projection-based decomposition (line 4).
By changing the scaling parameters $\omega_{\text{interf}}, \omega_{\text{noise}}, \omega_{\text{artif}}$ (lines 5, 6, and 7), we obtain the modified enhanced signals $\widehat{\mathbf{s}}_{\text{DSA}}$ with various amounts of each error component (line 8).
By inputting such modified enhanced signals to the ASR system (line 9) and evaluating the ASR performance (line 10), we can directly measure the impact of each error component on ASR performance.

Since DSA requires having access to the reference signals, it cannot be applied to practical use scenes, i.e., real recordings.
Here, we aim to identify how each error component impacts ASR performance.

\subsection{Improving SAR with observation adding post-processing}
\label{sec:oa}

As a practical approach to mitigating artifact errors, we consider the OA post-processing.
In this paper, we propose an interpretation of OA from the viewpoint of orthogonal projection-based decomposition and also give a mathematical proof that OA can monotonically increase the SAR under a mild condition.

\vskip0.5\baselineskip
\subsubsection{Formulation}

OA is a simple technique that interpolates the enhanced signal $\widehat{\mathbf{s}}$ with the observed signal $\mathbf{y}$ as\footnote{An alternative implementation of the OA post-processing consists of simply adding back a scaled version of the observed signal to the enhanced signal \cite{iwamoto2022how} as: $\widehat{\mathbf{s}}_{\text{OA}} = \widehat{\mathbf{s}} + \omega_{\text{OA}}\ \mathbf{y}$. These two implementations are equivalent except for the scale of the resultant modified signals. Note that the scale difference does not affect the orthogonal projection-based SE metrics (i.e., Eqs.~\eqref{eq:sdr}--\eqref{eq:sar}).}:
\begin{align}
    \widehat{\mathbf{s}}_{\text{OA}} &= (1-\omega_{\text{OA}})\ \widehat{\mathbf{s}} + \omega_{\text{OA}}\ \mathbf{y},\label{eq:oa_int}
\end{align}
%
where $\widehat{\mathbf{s}}_{\text{OA}} \in \mathbb{R}^{T}$ denotes the modified enhanced signal with OA, and $0 < \omega_{\text{OA}} < 1$ is an interpolation parameter that controls the balance between the enhanced and observed signals.
Note that setting $\omega_{\text{OA}} = 0$ gives the original enhanced signal, and $\omega_{\text{OA}} = 1$ the observed noisy signal.

Figure~\ref{fig:decomposition}-(b) illustrates the OA's behavior from the viewpoint of error decomposition.
In the figure, $\overline{\mathbf{s}}_{\text{target}}$, $\overline{\mathbf{e}}_{\text{noise}}$, and $\overline{\mathbf{e}}_{\text{artif}}$ denote the target source, noise error, and artifact error components of the modified enhanced signal with OA, respectively.
OA increases the target source and noise error components but decreases the artifact error component.

Unlike DSA, OA does not require having access to the reference signals and thus can be applied to real recordings.

\vskip0.5\baselineskip
\subsubsection{Theoretical analysis}

We show in the next proposition that OA improves the SAR of the enhanced signal $\widehat{\mathbf{s}}$ under a mild condition.
\begin{prop}
\label{prop:OA}
The OA operation in Eq~\eqref{eq:oa_int} improves the SAR of the original enhanced signal $\widehat{\mathbf{s}} = \mathrm{SE}(\mathbf{y})$ if 
it satisfies $\langle \widehat{\mathbf{s}}, \mathbf{y} \rangle > 0$.
\end{prop}
\begin{proof}
Based on Eqs.~\eqref{eq:target}--\eqref{eq:artif}, \eqref{eq:sar}, and \eqref{eq:oa_int}, the SAR improvement (denoted as $\text{SARi}$) from the original enhanced signal $\widehat{\mathbf{s}}$ to the modified enhanced signal with OA $\widehat{\mathbf{s}}_{\text{OA}}$ can be computed as:
\begin{align}
    \nonumber
    \text{SARi} 
    & = 
        10 \log_{10} \frac{\| \mathbf{P}_{\mathbf{s}, \mathbf{i}, \mathbf{n}} \widehat{\mathbf{s}}_{\text{OA}} \|^{2}}{\| \widehat{\mathbf{s}}_{\text{OA}} - \mathbf{P}_{\mathbf{s},\mathbf{i},\mathbf{n}} \widehat{\mathbf{s}}_{\text{OA}} \|^{2}} \\
        \nonumber
        & \hspace{7mm} - 10 \log_{10} \frac{\| \mathbf{P}_{\mathbf{s}, \mathbf{i}, \mathbf{n}} \widehat{\mathbf{s}} \|^{2}}{\| \widehat{\mathbf{s}} - \mathbf{P}_{\mathbf{s},\mathbf{i},\mathbf{n}} \widehat{\mathbf{s}} \|^{2}},
    \\
    \nonumber
    & = 
        10 \log_{10} \frac{\| (1-\omega_{\text{OA}})\ \mathbf{P}_{\mathbf{s}, \mathbf{i}, \mathbf{n}} \widehat{\mathbf{s}} + \omega_{\text{OA}}\ \mathbf{y} \|^{2}}{(1-\omega_{\text{OA}})^{2}\ \| \widehat{\mathbf{s}} - \mathbf{P}_{\mathbf{s}, \mathbf{i}, \mathbf{n}} \widehat{\mathbf{s}} \|^{2}}
        \\
        \nonumber
        & \hspace{7mm} - 10 \log_{10} \frac{\| \mathbf{P}_{\mathbf{s}, \mathbf{i}, \mathbf{n}} \widehat{\mathbf{s}} \|^{2}}{\| \widehat{\mathbf{s}} - \mathbf{P}_{\mathbf{s},\mathbf{i},\mathbf{n}} \widehat{\mathbf{s}} \|^{2}},
    \\
    \nonumber
    & \hspace{-5mm} =
        10 \log_{10} \frac{\| (1-\omega_{\text{OA}})\ \mathbf{P}_{\mathbf{s}, \mathbf{i}, \mathbf{n}} \widehat{\mathbf{s}} + \omega_{\text{OA}} \mathbf{y} \|^{2}}{(1-\omega_{\text{OA}})^{2}\ \| \mathbf{P}_{\mathbf{s}, \mathbf{i}, \mathbf{n}} \widehat{\mathbf{s}} \|^{2}},
    \\
    \nonumber
    & \hspace{-5mm} =
        10 \log_{10}
        \left[ 
            1
            + 
            \frac{
                \omega_{\text{OA}}^{2} \| \mathbf{y} \|^2 + 2 (1-\omega_{\text{OA}}) \omega_{\text{OA}} \langle \mathbf{P}_{\mathbf{s}, \mathbf{i}, \mathbf{n}} \widehat{\mathbf{s}}, \mathbf{y} \rangle
            }{
                (1-\omega_{\text{OA}})^{2}\ \| \mathbf{P}_{\mathbf{s}, \mathbf{i}, \mathbf{n}} \widehat{\mathbf{s}} \|^2
            }
        \right],
\end{align}
where we use $\mathbf{P}_{\mathbf{s}, \mathbf{i}, \mathbf{n}} \mathbf{y} = \mathbf{y}$ in the second equality.
Due to $0 < \omega_{\text{OA}} < 1$ from OA's definition, $(1-\omega_{\text{OA}}) \omega_{\text{OA}} > 0$.
Thus, $\text{SARi} > 0$ holds if
$\langle \mathbf{P}_{\mathbf{s}, \mathbf{i}, \mathbf{n}} \widehat{\mathbf{s}}, \mathbf{y} \rangle > 0$.
This sufficient condition can be rewritten as
$
\langle \mathbf{P}_{\mathbf{s}, \mathbf{i}, \mathbf{n}} \widehat{\mathbf{s}}, \mathbf{y} \rangle
= 
\langle \widehat{\mathbf{s}}, \mathbf{P}_{\mathbf{s}, \mathbf{i}, \mathbf{n}} \mathbf{y} \rangle
=
\langle \widehat{\mathbf{s}}, \mathbf{y} \rangle  > 0
$,
which concludes the proof.
\end{proof}

The condition $\langle \widehat{\mathbf{s}}, \mathbf{y} \rangle > 0$ holds with 50 \% probability even in the case where $\widehat{\mathbf{s}}$ is a randomly generated vector.
Thus, it should be met even for SE front-ends that are not very accurate, and even more likely to be met for the recent high-quality SE systems.
Assuming such SE systems, we can roughly expect $\widehat{\mathbf{s}} \approx \mathbf{s} + \varepsilon_{i} \mathbf{i} + \varepsilon_{n} \mathbf{n} + \widehat{\mathbf{e}}_{\text{artif}}$ with $\varepsilon_{n} \in \mathbb{R}$ and $\varepsilon_{i} \in \mathbb{R}$, where $|\varepsilon_{n}|$ and $|\varepsilon_{i}|$ are small.
Since ${\mathbf{e}}_{\text{artif}}$ is orthogonal to $\mathbf{s}$, $\mathbf{i}$, and $\mathbf{n}$ (i.e., $\langle \widehat{\mathbf{e}}_{\text{artif}}, \mathbf{s} \rangle = 0$, $\langle \widehat{\mathbf{e}}_{\text{artif}}, \mathbf{i} \rangle = 0$, and $\langle \widehat{\mathbf{e}}_{\text{artif}}, \mathbf{n} \rangle = 0$) and $\mathbf{s}$, $\mathbf{i}$, and $\mathbf{n}$ are generally uncorrelated (i.e., $\langle \mathbf{s}, \mathbf{i} \rangle \approx \mathbf{0}$, $\langle \mathbf{s}, \mathbf{n} \rangle \approx \mathbf{0}$, and $\langle \mathbf{i}, \mathbf{n} \rangle \approx \mathbf{0}$),
we have
$
\langle \widehat{\mathbf{s}}, \mathbf{y} \rangle
\approx 
\langle \mathbf{s} + \varepsilon_{i} \mathbf{i} + \varepsilon_{n} \mathbf{n} + \widehat{\mathbf{e}}_{\text{artif}}, \mathbf{s} + \mathbf{i} + \mathbf{n} \rangle
\approx 
\| \mathbf{s} \|^2 + \varepsilon_{i} \| \mathbf{i} \|^2 + \varepsilon_{n} \| \mathbf{n} \|^2
$.
Since $|\varepsilon_{i}|$ and $|\varepsilon_{n}|$ can be considered sufficiently small, we can expect  
$\langle \widehat{\mathbf{s}}, \mathbf{y} \rangle > 0$
and that OA almost always improves the SAR by Proposition~\ref{prop:OA}.

In addition to the above theoretical analysis, we experimentally investigated the effect of OA.
In our preliminary experiment (Section~\ref{sec:exp_oa_single}), we observed that the SAR scores of the modified enhanced signals with OA $\widehat{\mathbf{s}}_{\text{OA}}$ are better than those of the original enhanced signals $\widehat{\mathbf{s}}$ for all of the evaluated utterances.
Thus, the OA can be used as simple and practical post-processing to improve the SAR of the enhanced signals.

\subsection{Improving SAR with artifact-boosted training loss}

As an alternative approach to mitigating artifact errors and consequently improving ASR performance, we can consider adopting a training objective that reduces artifact errors.
Like OA (but unlike DSA), this scheme assumes that only the observed signal is available in the inference stage and thus can be applied to real recordings.

We assume that paired data of observed noisy signal and reference source, interference, and noise signals $\{\mathbf{y}, \mathbf{s}, \mathbf{i}, \mathbf{n}\}$ are available for training the SE model.
In this paper, as a novel training objective, we propose the AB-SDR loss as:
\begin{align}
    \mathcal{L}_{\text{AB-SDR}} = - 10 \log_{10} \frac{\| \mathbf{s}_{\text{target}} \|^{2}}{\| \mathbf{e}_{\text{interf}} + \mathbf{e}_{\text{noise}} + \alpha \mathbf{e}_{\text{artif}} \|^{2}}, \label{eq:ba_sdr}
\end{align}
%
where $\alpha > 1.0$ denotes the weighting parameter that prioritizes the artifact error component.
When we set $\alpha = 1.0$, the loss function corresponds to the original SDR loss \cite{vincent2006performance,boeddeker2021convolutive} that equally penalizes each error component (i.e., interference, noise, and artifact errors).
By setting $\alpha$ to a larger value, we boost the effect of artifact errors in the training loss, which drives the model to reduce artifact errors in the estimated signals more than with conventional losses.

An intuitive approach to improving the SAR would be to incorporate regularization for decreasing SAR (i.e., Eq.~\eqref{eq:sar}) in the training objective.
However, in our preliminary experiment, we observed that SAR-based regularization may cause training instability, probably because the SAR becomes optimal when the estimated signal matches the observed noisy signal.
Thus, in this study, we adopt an SDR-based loss that weighs the artifact error component as in Eq.~\eqref{eq:ba_sdr}.

\section{Experimental conditions}
\label{sec:exp_cond}

To confirm our hypothesis, i.e., the artifact errors have a larger impact on the degradation of ASR performance, we conducted the analytical experiments based on the proposed DSA (Section~\ref{sec:exp_dsa}). 
Then, we evaluated the effect of the proposed approaches to reducing the artifact errors, i.e., OA (Section~\ref{sec:exp_oa}) and AB-SDR (Section~\ref{sec:exp_tl}).
Moreover, we conducted additional analyses in Section~\ref{sec:exp_add}: 1) analysis of the effect of SE architecture and ASR's training data, 2) evaluation with different speech and noise corpora, and 3) evaluation with the real recordings.

\subsection{Evaluated datasets}
\label{sec:eval_data}

We created three datasets of simulated single-channel speech signals under reverberant and noisy conditions.
We used the Wall Street Journal (WSJ0) corpus \cite{paul1992design} and the Corpus of Spontaneous Japanese (CSJ) corpus \cite{maekawa2000spontaneous} for the speech signals, and used the CHiME-3 corpus \cite{barker2015third} and the DEMAND corpus \cite{thiemann2013diverse} for the noise signals.
We summarize the basic information of the evaluated datasets in
Table~\ref{tab:data}.

For the WSJ0 speech corpus, the speech sources for the training set were selected from the WSJ0's training set ``si\_tr\_s.''
Those for the development and evaluation sets were selected from the WSJ0's development set ``si\_dt\_05'' and evaluation set ``si\_et\_05.''
For the CSJ speech corpus, the speech sources for the training and development sets were selected from the CSJ's training set.
Those for the evaluation set were selected from the CSJ's evaluation sets ``eval1,'' ``eval2,'' and ``eval3.''
Here, we divided the CSJ's training set into subsets for training and development, containing 95 \% and 5 \%, respectively.

For the CHiME-3 noise corpus, we divided the noise sources into three subsets for training, development, and evaluation, containing 80 \%, 10\%, and 10\%, respectively, of the noise data of each environment (i.e., on a bus (BUS), in a cafe (CAF), a pedestrian area (PED), and at a street junction (STR)).
We used the fifth-channel signals of CHiME-3's noise data for generating noisy speech signals.
Similarly, for the DEMAND  noise corpus, we divided noise sources into three subsets for training, development, and evaluation, containing 70 \%, 10\%, and 20\%, respectively, of the noise data of each noise environment (e.g., living room (DLIVING), university restaurant (PRESTO), meeting room (OMEETING), etc.).
We used the first-channel signals of DEMAND's noise data for generating noisy speech signals.

To create the reverberant speech sources, we generated room impulse response (RIR) based on the image method \cite{allen1979image}, using the gpuRIR simulation toolkit \cite{diaz2021gpurir}.
We randomly generated the RIR configuration (e.g., room geometry, source and array positions, reverberation time, etc.) for each simulated RIR.
We set the reverberation time (T60) between 0.2 s and 0.6 s.

\begin{table}[t]
  \renewcommand{\arraystretch}{1.0}
  \caption{Summary of evaluated datasets.}
  \label{tab:data}
  \centering
  \scalebox{1.0}{
  \begin{threeparttable}[h]
  \begin{tabular}{ l | c c c }
    \toprule
    WSJ\_CHIME\_ST & \multirow{2}{*}{\# of utt.} & \multirow{2}{*}{SNR [dB]} & \multirow{2}{*}{SIR [dB]} \\
    (single-talker scenario) \\
    \midrule
    Training set & 30,000 & 0 $\sim$ 10 & -- \\
    Development set & 5,000 & 0 $\sim$ 10 & -- \\
    Evaluation set & 5,000 & 0 & -- \\
    \bottomrule
    \vspace{-1.5mm} \\
    \toprule
    WSJ\_CHIME\_MT & \multirow{2}{*}{\# of utt.} & \multirow{2}{*}{SNR [dB]} & \multirow{2}{*}{SIR [dB]} \\
    (multi-talker scenario) \\
    \midrule
    Training set & 30,000 & 5 $\sim$ 20 & 5 $\sim$ 20 \\
    Development set & 5,000 & 5 $\sim$ 20 & 5 $\sim$ 20 \\
    Evaluation set & 5,000 & 10 & 5 \\
    \bottomrule
    \vspace{-1.5mm} \\
    \toprule
    CSJ\_DEMAND\_ST & \multirow{2}{*}{\# of utt.} & \multirow{2}{*}{SNR [dB]} & \multirow{2}{*}{SIR [dB]} \\
    (single-talker scenario) \\
    \midrule
    Training set & 154,147 & 0 $\sim$ 10 & -- \\
    Development set & 8,112 & 0 $\sim$ 10 & -- \\
    Evaluation set & 3,949 & 0 & -- \\
    \bottomrule
    \vspace{-1.5mm} \\
    \toprule
    CHiME-3's real recordings & \multirow{2}{*}{\# of utt.} & \multirow{2}{*}{SNR [dB]} & \multirow{2}{*}{SIR [dB]} \\
    (single-talker scenario) \\
    \midrule
    Evaluation set & 1,320 & 2 $\sim$ 7\tnote{*} & -- \\    
    \midrule    
  \end{tabular}
  \begin{tablenotes}
    \item[*] The SNR value is estimated \cite{barker2017third} because we do not have access to reference source and noise signals for real recordings.
  \end{tablenotes}
  \end{threeparttable}
  }
\end{table}

\subsubsection{WSJ\_CHIME\_ST and WSJ\_CHIME\_MT}

We created single-talker and multi-talker noisy speech datasets based on the WSJ0's speech and CHiME-3's noise signals; we refer to the former single-talker dataset as ``WSJ\_CHIME\_ST'' and to the latter multi-talker dataset as ``WSJ\_CHIME\_MT.''
The WSJ\_CHIME\_ST dataset is similar to CHiME-3's simulation dataset but with different RIRs.
The original CHiME-3 simulation data contains only single-talker conditions and does not include the reference source and noise signals.
To create multi-talker data, we thus generated RIRs with the image method simulating far-field conditions.
We used these simulation configurations for the single-talker case as well to match the multi-talker case.

For the multi-talker scenario, the generated data consisted of simulated reverberant noisy two-speaker mixtures.
In this experiment, to avoid target speaker ambiguity, we assumed that the target speaker's speech is dominant within the utterance (i.e., the target speaker's speech has a larger power than the interfering speaker's speech).
Thus, the problem consists of extracting the dominant speaker\footnote{There also exists the source separation or target speech extraction setup, which separates all of the sources or extracts a single target source given the speaker enrollment.
In this study, for model simplicity, we focused on the dominant speech extraction setup, which would be one of the practical use cases for multi-talker ASR tasks.} in the two-speaker mixture.

\subsubsection{CSJ\_DEMAND\_ST}

To verify whether the experimental findings hold for different evaluated datasets, we also created another single-talker noisy speech dataset based on CSJ's speech and DEMAND's noise signals, which we refer to as ``CSJ\_DEMAND\_ST.''
``WSJ\_CHIME\_ST'' is created using the WSJ speech (English read speech) and CHiME-3 noise (4 environments) corpora, while ``CSJ\_DEMAND\_ST'' is created using the CSJ speech (Japanese Spontaneous speech) and DEMAND noise (18 environments) corpora.
It contains many more training utterances and a different language, more natural speaking styles, and more diverse noisy environments.

\subsubsection{CHiME-3's real recordings}

We used the simulated datasets to analyze the relations between the SE metrics (SDR, SIR, SNR, SAR \cite{vincent2006performance}) and the ASR metric (WER).
In addition, to confirm the effectiveness of our findings for real recordings, we used the real-recorded evaluation data of the CHiME-3 corpus ``et05\_real''.
Compared to the simulated evaluation sets, the real recordings comprise real reverberation and also cover more variations in terms of SNR condition (2 $\sim$ 7 dB \cite{barker2017third}).

\subsection{Evaluated systems}

\subsubsection{Speech enhancement front-end}

\begin{table}[t]
  \renewcommand{\arraystretch}{1.0}
  \caption{Summary of configurations of speech enhancement systems.}
  \label{tab:network}
  \centering
  \scalebox{1.0}{
  \begin{tabular}{ l r }
    \toprule
    Configuration of temporal convolutional network \\
    \midrule
    Number of channels in bottleneck (B) & 128 \\
    Number of channels in skip-connection (Sc) & 128 \\
    Number of channels in convolution blocks (H) & 512 \\
    Kernel size in convolution blocks (P) & 3 \\
    Number of convolution blocks in each repeat (X) & 8 \\
    Number of repeats (R) & 3\\
    \bottomrule
    \vspace{-1.5mm} \\
    \toprule
    Configuration of 1D convolution encoder/decoder \\
    \midrule
    Filter length & 16 samples \\
    Hop size & 8 samples \\
    Number of filters & 512 \\
    \bottomrule
    \vspace{-1.5mm} \\
    \toprule
    Configuration of STFT encoder/decoder \\
    \midrule
    Frame length & 1024 samples \\
    Frame shift & 256 samples \\
    Window function & Hanning \\
    \bottomrule
  \end{tabular}
  }
\end{table}

For the SE front-end, we adopted two types of SE models that work in different input and output domains (i.e., time-domain and short-time Fourier transform (STFT)-domain), representing major categories of SE model architectures.

We employed the temporal convolutional network (TCN)-based single-channel SE architecture \cite{kinoshita2020improving} that is similar to the time-domain audio separation network (TasNet) \cite{luo2019conv}.
The network is composed of an encoder layer, internal TCN blocks, and a decoder layer.
First, the encoder layer maps the time-domain signals to an intermediate representation, and then it is further processed by the convolution blocks.
Finally, the decoder layer directly remaps the intermediate representation to time-domain signals.

For the encoder/decoder mapping, we adopted two types of transformations: 1) learnable transformation with 1D convolution and 1D deconvolution and 2) fixed transformation with STFT and inverse short-time Fourier transform (iSTFT).
The former learnable transformation-based model (Enhan\_T) corresponds to the famous time-domain SE model, i.e., Conv-TasNet \cite{luo2019conv}, and the latter fixed transformation model (Enhan\_F) corresponds to its STFT-domain variant \cite{kinoshita2020improving}.
The configurations related to the encoder/decoder and TCN architecture are briefly summarized in Table~\ref{tab:network}, where we follow the notations introduced in a previous study \cite{luo2019conv}.

For the training loss, we compared three types of loss functions: 1) classical scale-dependent SNR loss \cite{roux2019sdr}, 2) original SDR loss ($\alpha = 1.0$ in Eq.~\eqref{eq:ba_sdr}), and 3) proposed AB-SDR loss ($\alpha > 1.0$ in Eq.~\eqref{eq:ba_sdr}).
We trained all of the models using the Adam optimizer \cite{kingma2015adam} with gradient clipping \cite{pascanu2013difficulty}, with an initial learning rate of 0.001, a batch size of 24, and a maximum of 100 epochs.
We set the number of basis vectors for the projections in the SDR and AB-SDR losses to $L=2$ and $L=1$ for the single-talker and multi-talker scenarios, respectively, based on the WER scores in the preliminary experiments\footnote{Note that setting $L$ to a larger value in the training loss provides too much degree of freedom for the projections, which allows the trained SE model to generate meaningless signals (e.g., high SDR but quite bad WER).}. 
In addition, for the AB-SDR loss, we tuned the weighting parameter $\alpha$ for values within $\{1.0, 1.5, \ldots, 3.5\}$, and obtained $\alpha= 1.5$ and $\alpha=2.0$ for single-talker and multi-talker scenarios, respectively. 

In this study, we implemented the SE systems based on Asteroid's Conv-TasNet implementation \cite{luo2019conv,pariente2020asteroid}.

\subsubsection{Speech recognition back-end}

The proposed DSA scheme requires a large number of ASR decoding experiments.
Thus for the ASR back-end, we adopted a computationally efficient deep neural network-hidden Markov model (DNN-HMM) hybrid ASR model\footnote{Parallel to this work, in another study \cite{iwamoto2023does}, we have investigated the effectiveness of OA post-processing for the state-of-the-art end-to-end ASR model \cite{chang2022end-to-end} with a pre-trained self-supervised learning-based feature extractor \cite{chen2022wavlm} trained using a huge amount of speech and noise data, and we confirmed the influence of artifact errors even for such a stronger huge ASR back-end.}.
The system was trained using the lattice-free maximum mutual information (MMI) criterion \cite{povey2016purely} and decoded with a trigram language model.
In this study, we implemented the ASR systems based on Kaldi’s standard recipe \cite{povey2011kaldi} for the CHiME-4 and CSJ tasks.
The details of the systems are shown in Kaldi's recipe\footnote{\url{https://github.com/kaldi-asr/kaldi/tree/master/egs/chime4}}$^, $\footnote{\url{https://github.com/kaldi-asr/kaldi/tree/master/egs/csj}}.

For the training dataset of the acoustic models, we adopted three types of signals: 1) noisy signal (Noisy), 2) enhanced signal (Enhan), and 3) noisy and enhanced signals (Noisy+Enhan).
We obtained the enhanced signals for the training dataset by inputting the noisy signals into the SE systems.
For the third option, we used both noisy and enhanced signals for the training, and thus the total number of training utterances (i.e., iterations) is twice those for the first and second options.

\subsection{Evaluation metrics}

As the evaluation metrics, we used six SE and one ASR measures; 1) SDR, 2) SIR, 3) SNR, 4) SAR, 5) STOI, 6) PESQ, and 7) WER, where higher is better for SDR, SIR, SNR, SAR, STOI, and PESQ, while lower is better for WER.
The number of basis vectors for calculating the orthogonal projection-based SE metrics (i.e., SDR, SIR, SNR, and SAR) is set to $L=512$ by following the de facto standard BSSEval’s implementation\footnote{\url{https://github.com/craffel/mir_eval}} \cite{vincent2006performance}.
Note that we slightly modified the BSSEval's implementation to accept reference noise signals $\mathbf{n}$ for computing the metrics under noisy conditions.

\section{Experimental analysis by directly scaling error components}
\label{sec:exp_dsa}

\subsection{Single-talker scenario}
\label{sec:exp_dsa_single}

\begin{figure}[t]
  \centering
  \includegraphics[width=1.0\linewidth]{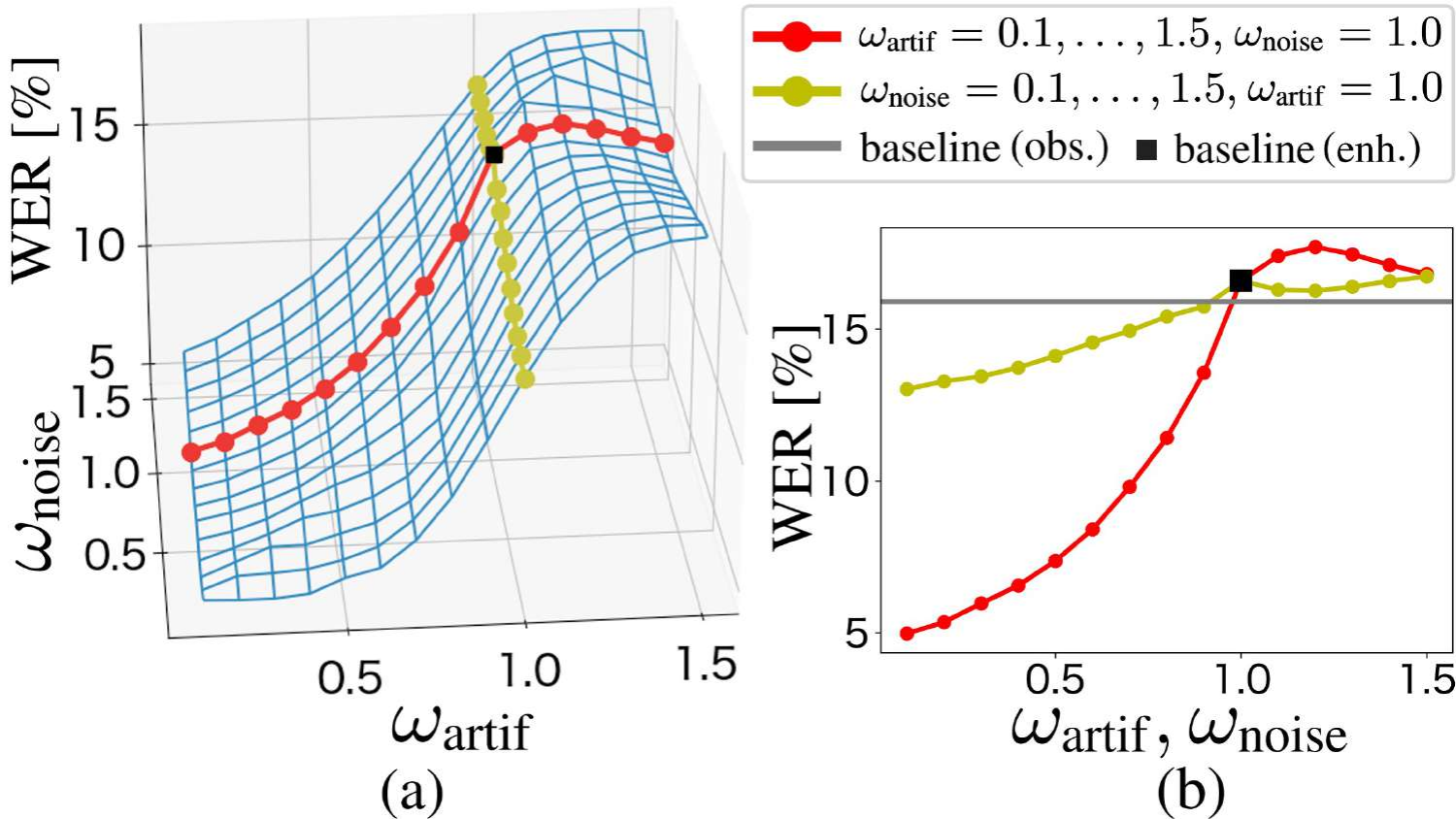}
  \caption{Results of DSA-based evaluation (WER [\%] (lower is better)) for single-talker setup (WSJ\_CHIME\_ST), which modifies the scale of noise ($\omega_{\text{noise}}$) and artifact ($\omega_{\text{artif}}$) error components.}
  \label{fig:dsa_single}
\end{figure}

Figure~\ref{fig:dsa_single} shows the results of DSA for the single-talker noisy speech dataset, WSJ\_CHIME\_ST.
Here, for the evaluated SE system, we adopt the time-domain SE model (Enhan\_T) trained with the SNR loss.
Note that in the single-talker scenario, the observed signal contains background noise but no interference speakers, and thus there are only two types of SE errors, i.e., noise and artifact\footnote{From Eq.~\eqref{eq:interf}, ${\mathbf{e}_{\text{interf}}}$ becomes zero because the orthogonal projector $\mathbf{P}_{\mathbf{s}, \mathbf{i}}$ is equal to $\mathbf{P}_{\mathbf{s}}$.}.
Figure~\ref{fig:dsa_single}-(a) is a three-dimensional plot showing the relationship between the SE errors (i.e., noise and artifact) and WER scores, which are obtained by the DSA procedure described in Section~\ref{sec:dsa} and Algorithm~\ref{alg:dsa}.
Figure~\ref{fig:dsa_single}-(b) is the corresponding two-dimensional plot, where we vary either the scaling parameters for noise $\omega_{\text{noise}}$ or artifact $\omega_{\text{artif}}$.
Each curve in the two-dimensional plot corresponds to the curve of the same color in the three-dimensional plot.

In the figures, the gray line (baseline (obs.)) represents the baseline score of the observed signals $\mathbf{y}$, and the black square (baseline (enh.)) represents the baseline score of the original enhanced signals $\widehat{\mathbf{s}}$, i.e., $\omega_{\text{noise}} = \omega_{\text{artif}} = 1$.
Similar to the results reported in the previous studies (e.g., \cite{yoshioka2015ntt,chen2018building,fujimoto2019one,menne2019investigation}), Figure~\ref{fig:dsa_single}-(b) shows that the WER score of the original enhanced signals (black square) is worse than that of the observed signals (gray line).

From the figures, we confirm that scaling down the artifact error component greatly improves ASR performance, while scaling down the noise error component has relatively little impact on ASR performance.
These results confirm our hypothesis that, among the two types of SE errors (i.e., noise and artifact), the artifact errors have a larger impact on the degradation of ASR performance.
This notable experimental finding suggests that the artifact error would be a reasonable signal-level numerical metric for explaining the cause of ASR performance degradation.

\subsection{Multi-talker scenario}
\label{sec:exp_dsa_multi}

\begin{figure}[t]
  \centering
  \includegraphics[width=1.0\linewidth]{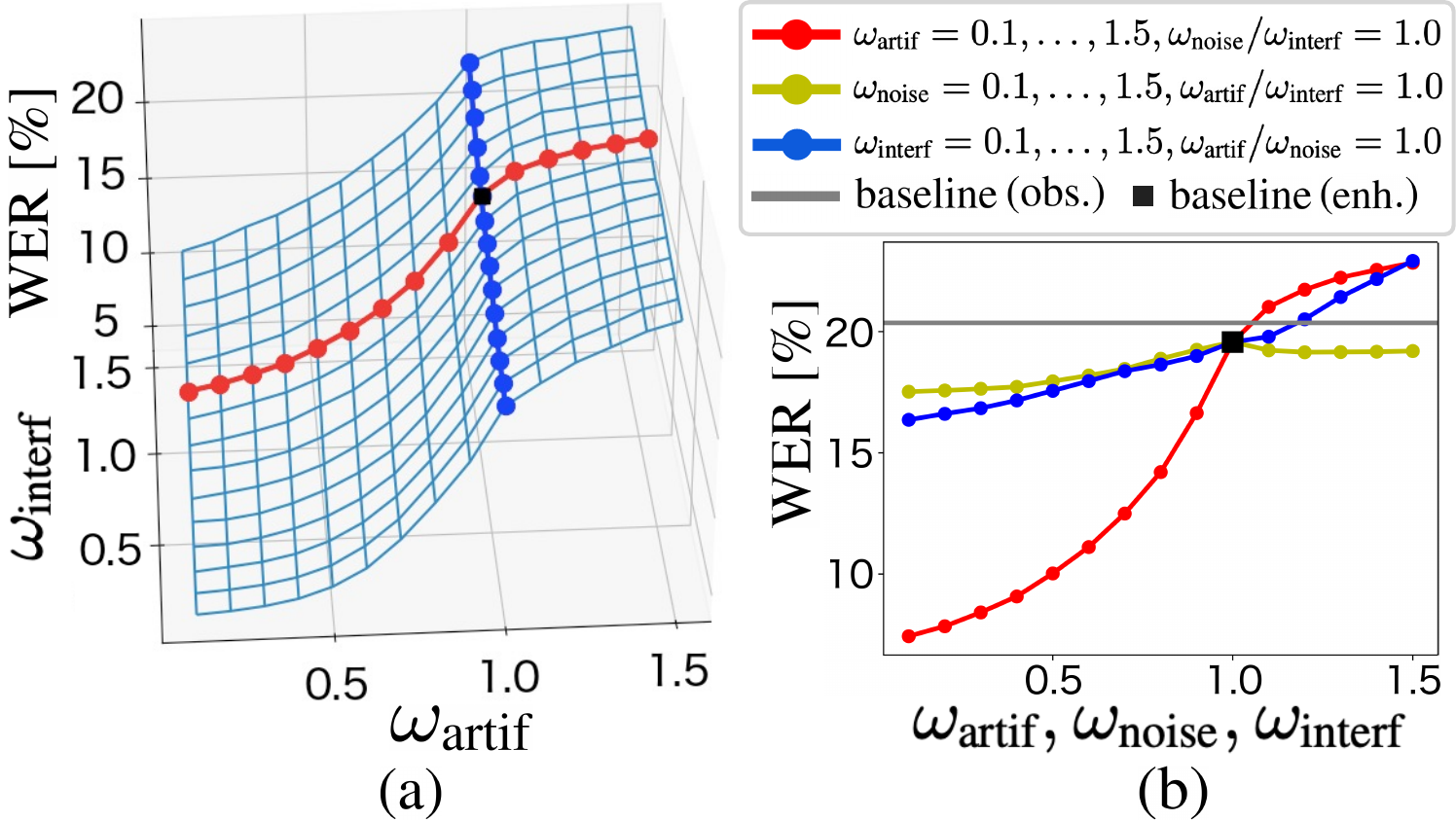}
  \caption{Results of DSA-based evaluation (WER [\%] (lower is better)) for multi-talker setup (WSJ\_CHIME\_MT), which modifies the scale of interference ($\omega_{\text{interf}}$), noise ($\omega_{\text{noise}}$), and artifact ($\omega_{\text{artif}}$) error components.}
  \label{fig:dsa_multi}
\end{figure}

Figure~\ref{fig:dsa_multi} shows the results of DSA for the multi-talker noisy speech dataset, WSJ\_CHIME\_MT.
Note that in the multi-talker scenario, the observed signal contains interfering speech and background noise, and thus there are three types of SE errors, i.e., interference, noise, and artifact.
Figure~\ref{fig:dsa_multi}-(a) and Figure~\ref{fig:dsa_multi}-(b) are three-dimensional and two-dimensional plots, respectively, showing the relationship between SE errors (i.e., interference, noise, and artifact) and WER scores.
Since we cannot illustrate a four-dimensional plot for the three scaling parameters, we instead show a three-dimensional plot in Figure~\ref{fig:dsa_multi}-(a), where we only vary the two scaling parameters of interference $\omega_{\text{interf}}$ and artifact $\omega_{\text{artif}}$.

From the figures, we confirm that scaling down the artifact error component greatly improves the ASR performance, while scaling down the noise or interference error component has relatively little impact on ASR performance.
This trend is similar to what has been observed in the single-talker scenarios (Section~\ref{sec:exp_dsa_single}).
These results demonstrate that our hypothesis (i.e., artifact errors have a larger impact on the degradation of ASR performance) would hold even for a multi-talker scenario, which would involve more general and challenging acoustic conditions.

\section{Experimental results with practical approaches mitigating artifact errors}
\label{sec:exp_oa_tl}

\subsection{Evaluation of observation adding post-processing}
\label{sec:exp_oa}

\subsubsection{Single-talker scenario}
\label{sec:exp_oa_single}

Figure~\ref{fig:result_oa_st} shows the results of the OA post-processing for the single-talker noisy speech dataset, WSJ\_CHIME\_ST.
Here, for the evaluated SE system, we adopt the time-domain SE model (Enhan\_T) trained with the SNR loss.
Figure~\ref{fig:result_oa_st}-(a) shows the SE metrics (i.e., SDR, SNR, and SAR) and Figure~\ref{fig:result_oa_st}-(b) shows the ASR metric (i.e., WER) of the enhanced signals modified with OA as described in Section~\ref{sec:oa}, where the interpolation parameter $\omega_{\text{OA}}$ varies within $\{0.0, 0.1, \ldots, 1.0\}$.
In the figures, the result of $\omega_{\text{OA}} = 0$ corresponds to the baseline score of the original enhanced signals $\widehat{\mathbf{s}}$ (baseline (enh.)), and that of $\omega_{\text{OA}} = 1$ corresponds to the baseline score of the observed signals $\mathbf{y}$ (baseline (obs.)).

The figure shows that when $\omega_{\text{OA}}$ increases (i.e., increasing the proportion of the observed signal in $\widehat{\mathbf{s}}_{\text{OA}}$), the SDR and SNR decrease while the SAR monotonically increases.
Moreover, we can confirm that, by adjusting $\omega_{\text{OA}}$ around $0.5$ (i.e., adding the observed signal to the original enhanced signal in similar proportion), the modified enhanced signals with OA achieve better WER than both the baseline observed signals (gray line) and the original enhanced signals (black square).

\subsubsection{Multi-talker scenario}
\label{sec:exp_oa_multi}

\begin{figure}[t]
  \centering
  \includegraphics[width=1.0\linewidth]{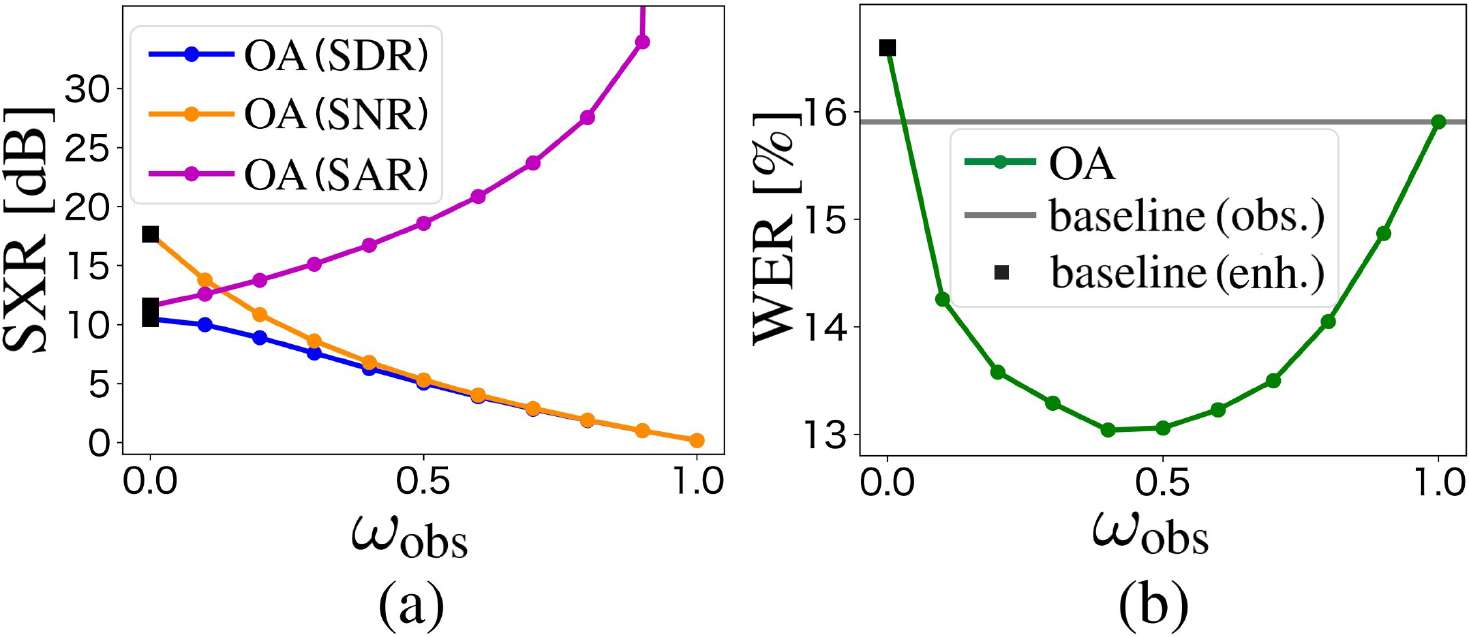}
  \caption{Results of OA post-processing (SDR, SNR, SAR [dB] (higher is better) and WER [\%] (lower is better)) for single-talker setup (WSJ\_CHIME\_ST). We use the compact notation SXR, in place of SDR, SNR, and SAR, to denote the ratio between two signals S and X, where X is either the distortion, noise, or artifacts.}
  \label{fig:result_oa_st}
\end{figure}
\begin{figure}[t]
  \centering
  \includegraphics[width=1.0\linewidth]{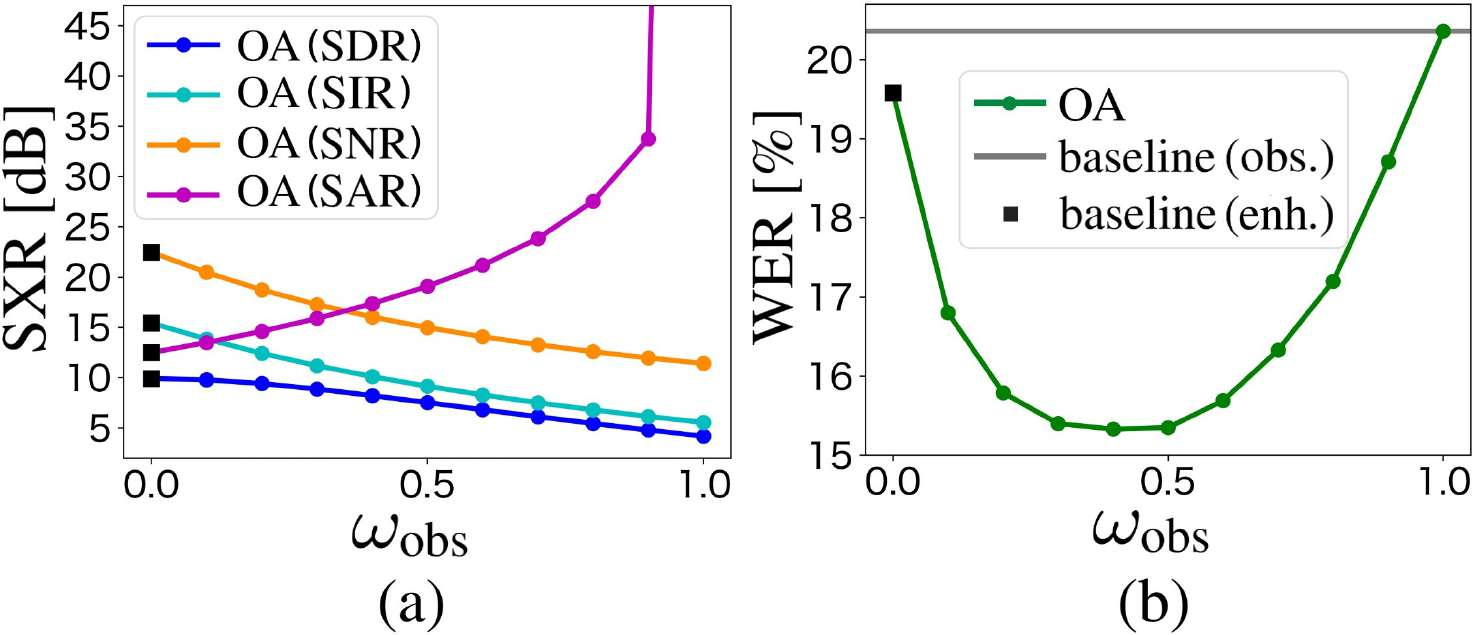}
  \caption{Results of OA post-processing (SDR, SIR, SNR, SAR [dB] (higher is better) and WER [\%] (lower is better)) for multi-talker setup (WSJ\_CHIME\_MT). SXR denotes values encompassing SDR, SIR, SNR, and SAR.}
  \label{fig:result_oa_mt}
\end{figure}

Figure~\ref{fig:result_oa_mt} shows the results of the OA post-processing for the multi-talker noisy speech dataset, WSJ\_CHIME\_MT.
Figure~\ref{fig:result_oa_mt}-(a) and Figure~\ref{fig:result_oa_mt}-(b) show the SE metrics (i.e., SDR, SIR, SNR, and SAR) and the ASR metric (i.e., WER) of the enhanced signals modified with OA, respectively.

In the results of the multi-talker scenario in Figure~\ref{fig:result_oa_mt}, we can observe a similar trend to the results of the single-talker scenario in Figure~\ref{fig:result_oa_st}: the OA post-processing 1) decreases the SDR, SIR, and SNR but monotonically increases the SAR and 2) achieves a better WER score than both the observed (gray line) and enhanced (black square) signals.

These results support our hypothesis that 1) the artifact component has a worse impact on ASR performance compared to the other error components and 2) the ASR performance of single-channel SE front-ends can be improved by decreasing the ratio of artifact errors (or increasing SAR).
Moreover, the results demonstrate that OA post-processing would be an effective and practical approach to mitigating the impact of artifact errors, not only in the single-talker but also in the multi-talker scenario.

\subsection{Evaluation of artifact-boosted training loss}
\label{sec:exp_tl}

\subsubsection{Single-talker scenario}
\label{sec:exp_tl_single}

\begin{table}[t]
  \renewcommand{\arraystretch}{1.0}
  \caption{Results of artifact-boosted training loss (SDR, SNR, SAR [dB], STOI, PESQ (higher is better) and WER [\%] (lower is better)) for single-talker setup (WSJ\_CHIME\_ST).}
  \label{tab:result_tl_st}
  \centering
  \scalebox{0.9}{
  \begin{tabular}{@{}c | wl{0.8cm} wl{0.9cm} wc{0.4cm} | wc{0.4cm} wc{0.4cm} wc{0.4cm} wc{0.4cm} wc{0.4cm} wc{0.4cm} wc{0.4cm}@{}}
    \toprule
    Id & System & Loss & $\omega_{\text{OA}}$ & SDR & SNR & SAR & WER & STOI & PESQ \\
    \midrule
    (1) & clean & -- & -- & -- & -- & -- & 4.4 & -- & -- \\
    (2) & noisy & -- & -- & 0.2 & 0.2 & -- & 15.9 & 0.73 & 1.21 \\
    \midrule
    (3) & Enhan\_T & SNR & -- & 10.5 & 17.7 & 11.6 & 16.6 & 0.87 & 2.00 \\
    (4) & & SDR & -- & 12.0 & 16.9 & 14.0 & 14.8 & 0.86 & 2.33 \\
    (5) & & AB-SDR & -- & 11.3 & 13.1 & 16.7 & 13.0 & 0.84 & 1.94 \\
    \midrule
    (6) & \ \ +OA & SNR & 0.4 & 6.3 & 6.8 & 16.7 & 13.0 & 0.83 & 1.53 \\
    (7) & & SDR & 0.4 & 4.1 & 4.4 & 18.8 & 13.1 & 0.77 & 1.25 \\
    (8) & & AB-SDR & 0.1 & 9.9 & 11.1 & 17.1 & \textBF{12.8} & 0.80 & 1.59 \\
    \bottomrule
  \end{tabular}
  }
\end{table}

Table~\ref{tab:result_tl_st} compares evaluated SE systems trained with the conventional SNR and SDR and the proposed AB-SDR losses in terms of SE (i.e., SDR, SNR, SAR, STOI, and PESQ) and ASR (i.e., WER) performance for the single-talker noisy speech dataset, WSJ\_CHIME\_ST.
We used the time-domain SE model (Enhan\_T) in this experiment.
In addition, Table~\ref{tab:result_tl_st} shows the performance of the evaluated SE systems applied with the OA post-processing (Enhan\_T+OA).
Here, we tuned the interpolation parameter $\omega_{\text{OA}}$ for each SE system by varying it within $\{0.0, 0.1, \ldots, 1.0\}$ and set the hyperparameter that achieved the best WER scores on the development set.

From the table, we observe that Enhan\_T trained with AB-SDR loss (row 5) achieved a significantly higher SAR score, at the expense of decreasing the SNR score than Enhan\_T trained with SNR and SDR losses (rows 3 and 4).
As expected, we could confirm that the AB-SDR loss (row 5) achieved a better WER score than both the SNR and SDR losses (rows 3 and 4) and the baseline observed signals (row 2).

\subsubsection{Multi-talker scenario}
\label{sec:exp_tl_multi}

\begin{table}[t]
  \renewcommand{\arraystretch}{1.0}
  \caption{Results of artifact-boosted training loss (SDR, SIR, SNR, SAR [dB], STOI, PESQ (higher is better) and WER [\%] (lower is better)) for multi-talker setup (WSJ\_CHIME\_MT).}
  \label{tab:result_tl_mt}
  \centering
  \scalebox{0.9}{
  \begin{tabular}{ @{}c | wl{0.8cm} wl{0.9cm} wc{0.4cm} | wc{0.4cm} wc{0.4cm} wc{0.4cm} wc{0.4cm} wc{0.4cm} wc{0.4cm} wc{0.4cm}}
    \toprule
    Id & System & Loss & $\omega_{\text{OA}}$ & SDR & SIR & SNR & SAR & WER & STOI & PESQ \\
    \midrule
    (1) & clean & -- & -- & -- & -- & -- & -- & 3.6 & -- & -- \\
    (2) & noisy & -- & -- & 4.2 & 5.5 & 11.4 & -- & 20.4 & 0.77 & 1.39 \\
    \midrule
    (3) & Enhan\_T & SNR & -- & 9.9 & 15.4 & 22.5 & 12.5 & 19.6 & 0.87 & 1.97 \\
    (4) & & SDR & -- & 9.9 & 15.6 & 22.7 & 12.4 & 19.1 & 0.87 & 1.99 \\
    (5) & & AB-SDR & -- & 8.5 & 10.6 & 16.6 & 17.1 & \textBF{14.9} & 0.85 & 1.75 \\
    \midrule
    (6) & \ \ +OA & SNR & 0.4 & 8.2 & 10.1 & 16.0 & 17.4 & 15.2 & 0.85 & 1.72 \\
    (7) & & SDR & 0.4 & 8.3 & 10.2 & 16.1 & 17.3 & \textBF{14.9} & 0.85 & 1.72 \\
    (8) & & AB-SDR & 0.0\tablefootnote{$\omega_{\text{OA}} = 0.0$ implies that the OA post-processing did not improve the WER scores.} & 8.5 & 10.6 & 16.6 & 17.1 & 14.9 & 0.85 & 1.75 \\
    \bottomrule
    
  \end{tabular}
  }
\end{table}

Table~\ref{tab:result_tl_mt} compares evaluated SE systems trained with the conventional SNR and SDR and the proposed AB-SDR losses in terms of SE and ASR performance for the multi-talker noisy speech dataset, WSJ\_CHIME\_MT.
We used the time-domain SE model (Enhan\_T) and also applied OA post-processing (Enhan\_T+OA) in this experiment.

The table shows that the AB-SDR loss decreases the SIR and SNR scores but monotonically increases the SAR score, and it achieves a better WER score compared with both the observed signals and enhanced signals trained with SNR and SDR losses.
This trend is similar to what has been observed in the single-talker scenario (in Table~\ref{tab:result_tl_st}).

These results further support our hypothesis (i.e., the artifact component has a worse impact on ASR performance) and demonstrate that it would hold not only for the single-talker but also for the multi-talker scenario.
In addition, the results demonstrate that designing the training objective while considering the artifact errors contributes to improving ASR performance with single-channel SE front-ends, regardless of whether the single-talker or multi-talker scenarios is used.

Moreover, Table~\ref{tab:result_tl_st} and Table~\ref{tab:result_tl_mt} show that Enhan\_T applied with both the AB-SDR loss and the OA post-processing (row 8) achieved the best WER score but the performance improvement was not as significant as when using OA with systems trained with SNR or SDR losses. 
This is probably because both approaches produce a similar effect on the resultant evaluated signals, i.e., decreasing the artifact (SAR) errors at the expense of increasing the interference (SIR) and noise (SNR) errors.

In our preliminary experiments, we investigated using the thresholded scale-invariant SAR \cite{roux2019sdr,koizumi2022snri} as a regularization term in addition to the SDR loss.
A careful examination of the results revealed that the SAR of the enhanced signal did not improve for some utterances but became very large for the other utterances. In the latter cases, the high SAR was achieved by generating enhanced signals that were very close to the observed signals $\mathbf{y}$.
Consequently, although the average SAR score improved, it did not lead to improved WER score.

\section{Additional analyses}
\label{sec:exp_add}

\subsection{Evaluation with different speech enhancement front-ends and speech recognition back-ends}

\begin{table}[t]
  \renewcommand{\arraystretch}{1.0}
  \caption{Results of OA post-processing with two SE front-ends and five ASR back-ends (WER [\%] (lower is better)) on WSJ\_CHIME\_ST dataset.}
  \label{tab:result_se_sr}
  \centering
  \scalebox{0.9}{
  \begin{tabular}{ l | c c c c c }
    \toprule
    \multirow{2}{*}{\diagbox{SE}{ASR}} & \multirow{2}{*}{Noisy} & \multirow{2}{*}{Enhan\_T} & \multirow{2}{*}{Enhan\_F} & Noisy &  Noisy \\
    & & & & +Enhan\_T & +Enhan\_F \\
    \midrule
    Noisy & 15.9 & 30.5 & 23.9 & 15.5 & 14.9  \\
    \midrule
    Enhan\_T & 16.6 & 12.9 & 13.5 & 13.7 & 14.8 \\
    \ + OA & 13.0 & 12.9 & 13.3 & \textBF{12.3} & \textBF{12.3} \\
    \midrule
    Enhan\_F & 26.3 & 22.8 & 15.0 & 21.0 & 16.1 \\
    \ + OA & 14.6 & 18.6 & 15.0 & 14.1 & 13.6 \\
    \bottomrule
    
  \end{tabular}
  }
\end{table}

In Section~\ref{sec:exp_oa_tl}, we conducted experiments with the time-domain SE front-end and the ASR back-end trained independently of the SE front-end, i.e., trained with reverberant noisy speech.
In this section, to investigate the influence of the SE model architecture (especially, input and output domains), we also adopt the two SE front-ends: 1) time-domain SE model (Enhan\_T) and 2) STFT-domain SE model (Enhan\_F).
In addition, to investigate the influence of the training data for the ASR back-end, we evaluated the behavior of OA post-processing on the five ASR back-ends, which are trained with 1) the noisy signals (Noisy), 2) the enhanced signals generated by the time-domain SE model (Enhan\_T), 3) the enhanced signals generated by the STFT-domain SE model (Enhan\_F), 4) noisy signals and enhanced signals generated by Enhan\_T, and 5) noisy signals and enhanced signals generated by Enhan\_F.

Table~\ref{tab:result_se_sr} shows the WER scores for the combination of the evaluated SE and ASR systems.
From the table, we can confirm that even when we use the weaker SE front-end (Enhan\_F), the OA post-processing improves the ASR performance.
These results show that mitigating the artifact errors is effective regardless of the SE model architecture.

Moreover, the table shows that by using the ASR back-end trained with the enhanced signals (e.g., Noisy+Enhan\_T), both of the SE front-ends (Enhan\_T and Enhan\_F) obtained the ASR performance improvements compared with the ASR back-end trained independently of the SE front-end (Noisy).
When the ASR back-end is trained exclusively on the matched SE front-end, OA is ineffective, as expected.
However, when the back-end is trained on the combination of noisy and enhanced speech, OA improves performance.
Interestingly, the latter system also achieved the best WER scores.
These results suggest that ASR performance may still suffer from artifact errors not only when using the ASR back-end trained independently of the SE front-end but also when using an ASR back-end specialized for an SE front-end (i.e., trained with the enhanced signals of a specific SE front-end).

\subsection{Evaluation with different speech and noise corpora}

\begin{figure}[t]
  \centering
  \includegraphics[width=1.0\linewidth]{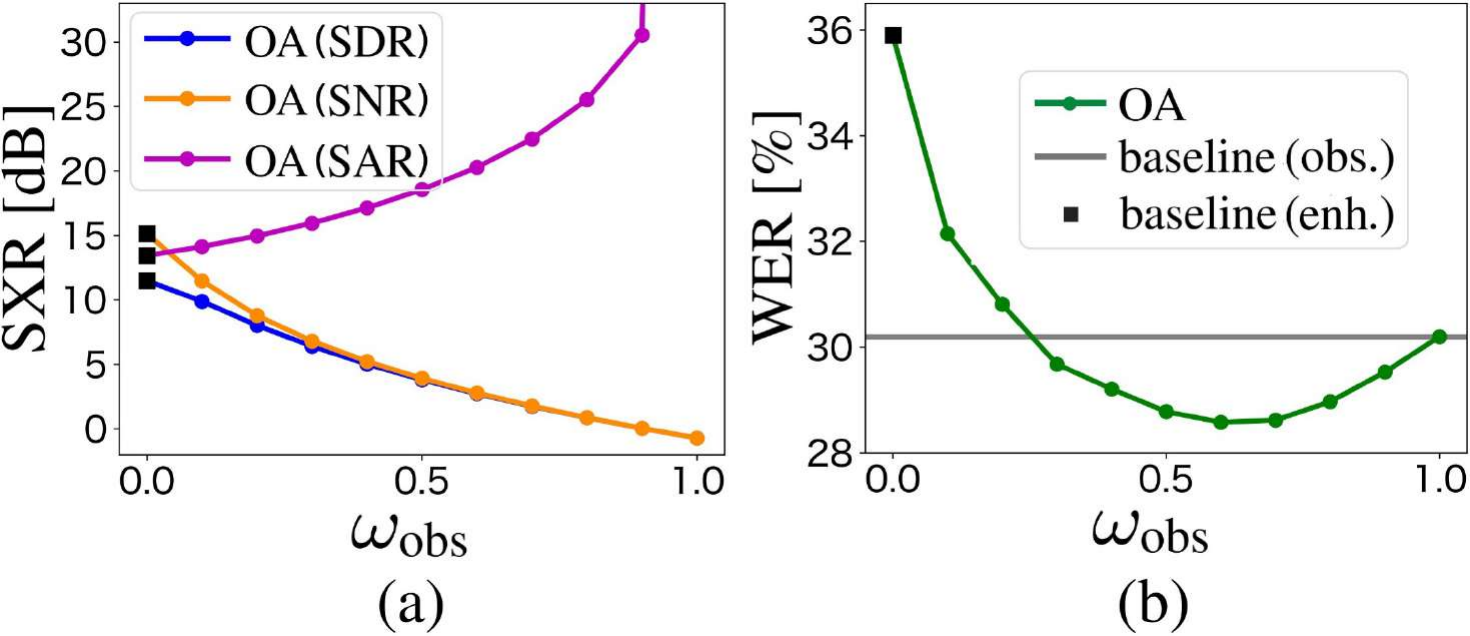}
  \caption{Results of OA post-processing (SDR, SNR, SAR [dB] (higher is better) and WER [\%] (lower is better)) on CSJ\_DEMAND\_ST dataset. SXR denotes values encompassing SDR, SNR, and SAR.}
  \label{fig:result_oa_st_csj_demand}
\end{figure}

To verify whether the findings obtained by the above experiments hold for different speech and noise conditions, we also evaluated the behavior of OA post-processing on another single-talker noisy speech dataset, CSJ\_DEMAND\_ST, which is created based on the CSJ speech (Japanese Spontaneous speech) and DEMAND noise (18 environments) corpora.
Figure~\ref{fig:result_oa_st_csj_demand}-(a) shows the SE metrics (i.e., SDR, SNR, and SAR), and Figure~\ref{fig:result_oa_st_csj_demand}-(b) shows the ASR metric (i.e., WER) of the enhanced signals modified with OA as described in Section~\ref{sec:oa}, where the interpolation parameter $\omega_{\text{OA}}$ varies within $\{0.0, 0.1, \ldots, 1.0\}$.

The figure shows a similar trend to the results on the WSJ\_CHIME\_ST dataset in Section~\ref{sec:exp_oa_single}: 1) by increasing $\omega_{\text{OA}}$, the SAR monotonically increases at the expense of decreasing the SNR and 2) by adjusting $\omega_{\text{OA}}$ around $0.5$, the modified enhanced signals with OA achieve better WER than both the baseline observed (gray line) and enhanced (black square) signals.
These results further increase the experimental evidence of our hypothesis and show that it would hold for different evaluated datasets.

\subsection{Evaluation with real recordings}

To confirm the effectiveness of our findings even for real recordings, we evaluated the OA post-processing and the artifact-boosted loss training for the real-recorded evaluation data, i.e., ``et05\_real'' of the CHiME-3 dataset.
Table~\ref{tab:result_oa_real} shows the WER scores of the evaluated SE systems trained with different loss functions (SNR, SDR, and AB-SDR) with/without OA post-processing (Enhan\_T and Enhan\_T+OA).
From the table, we observe that both OA post-processing and artifact-boosted loss training, which were shown to improve the SAR (see Sections~\ref{sec:exp_oa} and \ref{sec:exp_tl}), are also effective at improving the WER even when they are applied to real recordings, e.g., by comparing (2) with (5) and comparing (2) with (4), respectively.
These results suggest that the findings based on the orthogonal projection-based analysis of the simulated dataset would hold for real recordings, i.e., mitigating the artifact errors would be a key to improving the ASR performance with the single-channel SE front-ends.

\begin{table}[t]
  \renewcommand{\arraystretch}{1.0}
  \caption{Results of OA post-processing and AB-SDR training loss (WER [\%] (lower is better)) for real recordings (CHiME-3's et05\_real).}
  \label{tab:result_oa_real}
  \centering
  \scalebox{1.0}{
  \begin{tabular}{ c | l l c | c }
    \toprule
    Id & System & Loss & $\omega_{\text{OA}}$ & WER \\
    \midrule
    (1) & noisy & -- & -- & 28.7 \\
    \midrule
    (2) & Enhan\_T & SNR & -- & 31.7 \\
    (3) & & SDR & -- & 30.8 \\
    (4) & & AB-SDR & -- & \textBF{22.3} \\
    \midrule
    (5) & \ \ +OA & SNR & 0.4 & 22.5 \\
    (6) & & SDR & 0.5 & 22.5  \\
    (7) & & AB-SDR & 0.0 & 22.3  \\
    \bottomrule
  \end{tabular}
  }
\end{table}

\section{Conclusion}
\label{sec:conclusion}

In this paper, we analyzed the processing distortion problem induced by single-channel SE that would limit the potential improvement of ASR performance.
We proposed a novel analysis scheme that applied orthogonal projection-based decomposition of SE errors and experimentally found a reasonable signal-level numerical metric, i.e., artifact errors, that could explain the cause of ASR performance degradation.

In addition, we explored two practical approaches to mitigating the effect of artifact errors for improving ASR performance, i.e., OA post-processing and AB-SDR training loss.
We proved that OA can monotonically increase the SAR value under mild conditions.
Furthermore, we experimentally confirmed that both approaches successfully reduced artifact errors and improved the ASR performance of single-channel SE, even for real recordings.
These results further support our finding that artifact errors have a larger impact on the degradation of ASR performance than residual interference and noise errors.

The signal-level interpretation and experimental findings of the processing distortion obtained in this paper will provide novel insights into designing better SE front-ends and open novel research directions for building single-channel noise-robust ASR.
Future works should include extension of this analysis to different front-end (e.g., dereverberation) and back-end (e.g., speaker recognition) tasks, as well as to multichannel scenarios.


\bibliographystyle{IEEEtran}
\bibliography{mybib}

\end{document}